\documentclass[conference]{IEEEtran}
\usepackage{graphicx}
\usepackage{subfigure}
\usepackage{amsmath}
\usepackage{theorem}
\usepackage{amssymb}
\usepackage{comment}
\usepackage{cite}
\usepackage{color}
\newtheorem{Theorem}{Theorem}
\newtheorem{lemma}{\textbf{Lemma}}
\newtheorem{Corollary}{Corollary}
\newtheorem{DD}{Definition}
\usepackage[]{algorithm2e}
\usepackage{fancyhdr}
\begin{document}
\title{Exchanging Third-Party Information with Minimum Transmission Cost}
\author{
\small Xiumin Wang$^\star$, Wentu Song$^\dagger$$^\star$, Chau Yuen$^\star$, Jing Li (Tiffany)$^\ddag$\\
$^\star$ Singapore University of Technology and Design, Singapore\\
$^\dagger$ School of Mathematical Sciences, Peking University, China\\
$^\ddag$ Department of Electrical and Computer Engineering, Lehigh University, Bethlehem, PA 18015\\
\small Email: wangxiumin@sutd.edu.sg, songwentu@gmail.com, yuenchau@sutd.edu.sg, jingli@ece.lehigh.edu
}
\maketitle


\begin{abstract}
In this paper, we consider the problem of minimizing the total transmission cost for exchanging channel state information. We proposed a network coded cooperative data exchange scheme, such that the total transmission cost is minimized while each client can decode all the channel information held by all other clients. In this paper, we first derive a necessary and sufficient condition for a feasible transmission. Based on the derived condition, there exists a feasible code design to guarantee that each client can decode the complete information. We further formulate the problem of minimizing the total transmission cost as an integer linear programming. Finally, we discuss the probability that each client can decode the complete information with distributed random linear network coding.

\end{abstract}

\section{Introduction}
In wireless networks, it is always beneficial for the nodes to know the global knowledge of channel state information (CSI), e.g., channel gain or link loss probability, since global information can greatly ease the network optimization and improve the performance.
Generally, such CSI on a connected link can be regarded as a local information and known between two connected nodes (e.g., node $i$ and $j$). However, for a {\em third-party} node, e.g., the node $k\neq i,j$, the channel information of link $(i,j)$ is unknown to it. In some network design, such third-party information communication \cite{Love2007,Aluko2010} is often necessary.

Recently, cooperative data exchange among the users \cite{ElRouayheb2010} becomes one of the most promising approaches in designing efficient data communications.
In cooperative data exchange, each client initially holds a subset of packets, which are broadcast from the server or locally generated by itself.
The objective is to guarantee that all the clients finally obtain the whole packets by cooperatively exchanging
the data. Recent works showed that network coding \cite{Ahlswede2000,Katti2006,Wang2012} can reduce the number of transmissions or delay required for cooperative data exchange. However, finding the optimal solution with network coding to minimize the number of transmissions \cite{ElRouayheb2010,Sprintson2010,Milosavljevic2011} or the transmission cost \cite{Ozgul2011,Tajbakhsh2011} is non-trivial for general cooperative data exchange problem.

The work in \cite{Love2007} designs a coded cooperative data exchange scheme to minimize the number of transmissions for third-party information exchange. Compared with general cooperative data exchange problem, in third-party information exchange, each client initially has the local CSI to all other connected clients, and wants to know all CSI knowledge that is unknown to it. The work in \cite{Love2007} showed an optimal transmission scheme to minimize the total number of transmissions for exchanging third-party information.

Although the work in \cite{Love2007} gives the optimal solution to minimize the total number of transmissions required for third-party information exchange, it does not consider the case where each client is associated with a transmission cost, as studied in \cite{Ozgul2011}.
Consider a three-client network as shown in Fig.~\ref{Fig.thirdparty}, where $x_{i,j}$ denotes the CSI between client $i$ and client $j$. It is assumed that the links are symmetric, i.e., $x_{i,j}=x_{j,i}$. Initially, client $i$ knows only the local information $x_{i,j}$ for $\forall j\neq i$. Without considering the cost, client $1$ and client $2$ may be selected to transmit the encoded packets $x_{1,2}+x_{1,3}$ and $x_{1,2}+x_{2,3}$, respectively, to complete the data exchange process. However, if we consider the cost given as $\delta_i$, selecting client $2$ and $3$ as the transmitters is a better choice than the former solution in terms of the total transmission cost.


\begin{figure}[h]
\begin{center}
\includegraphics[height=30mm,width=72mm]{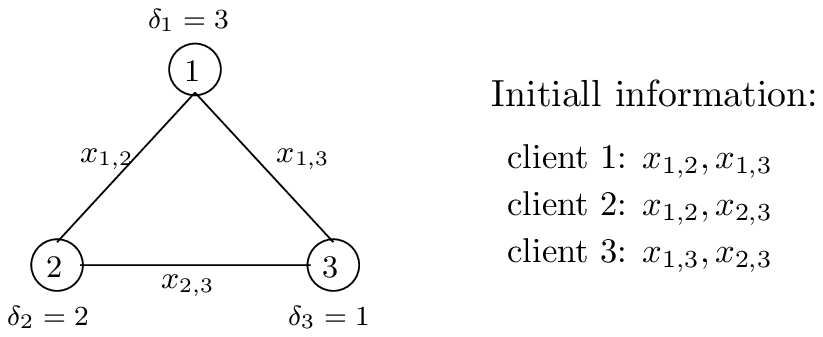}
\caption{Third-party information exchange among three clients}\label{Fig.thirdparty}
\end{center}
\end{figure}

In this paper, we design an algorithm to determine the number of packets that each client should send and how the packets should be encoded for each transmission, so as to minimize the total transmission cost required for the third-party information exchange problem. Similar to the previous works \cite{ElRouayheb2010, Sprintson2010, Love2007}, we consider there is a common control channel which allows reliable broadcast by any client to all the other clients. The main contributions of this paper can be summarized as follows:
\begin{itemize}
\item We derive a necessary and sufficient condition for a feasible transmission scheme such that there exists a code design for every client to successfully decode all the packets from other clients.
\item Based on the necessary and sufficient condition for feasible transmission, we formulate the problem of minimizing the total transmission cost as an integer linear programming.
\item Our analysis shows that the clients with lower transmission costs should send more packets than the clients with higher transmission costs.
\item We analyze the probability that every client can decode all other packets when random linear network coding is locally performed at each client.
\end{itemize}

The rest of the papers are organized as follows. The problem is formulated in Section~\ref{Sec.assumption}. Section~\ref{Sec.solution} derives the necessary and sufficient condition for a feasible transmission scheme. In Section~\ref{Sec.cost}, we give the optimal solution with the minimum transmission cost and analyze the performance with random network coding. We conclude the paper in Section~\ref{Sec.conclusion}.

\section{Problem Formulation}\label{Sec.assumption}

Consider a network with $N$ clients in $C=\{c_1,c_2,\cdots, c_N\}$, where each client $c_i\in C$ is associated with a transmission cost $\delta_i$ for sending a single packet.
Suppose that $x_{i,j}$ is the CSI (e.g., channel gain or link loss probability) of the link between client $c_i$ and client $c_j$. Initially, each client $c_i$ only knows the local CSI, i.e., client $c_i$ only holds the packets in $X_i=\{x_{i,j}|\forall j\in\{1,2,\cdots,N\}\setminus \{i\}\}$.
We assume that the links are symmetric, i.e., $x_{i,j}=x_{j,i}$ for $\forall i,j$. In other words, for every two clients $c_i$ and $c_{j}$, they hold one common packet $x_{i,j}$. Thus,
the set of all the packets is $X=\{x_{1,2}, \cdots, x_{1,N}, x_{2,3}, \cdots, x_{2,N}, \cdots, x_{N-1,N}\}$. Suppose that $K$ is the total number of the packets in the network, i.e., $K=|X|=\frac{N(N-1)}{2}$.



There is a lossless broadcast channel for clients to send or receive the packets \cite{ElRouayheb2010,Love2007,Sprintson2010,Milosavljevic2011,Ozgul2011,Tajbakhsh2011,MuxiYan2011}. Each transmitted packet is encoded over the packets initially held by the sender.
Let $y_i$ be the number of packets required to be transmitted by client $c_i$. The total transmission cost can thus be written as
\begin{eqnarray}\label{objective.cost}
\delta=\sum_{i=1}^N\delta_iy_i
\end{eqnarray}In this paper, our goal is to find a network coded transmission scheme that satisfies the following two conditions:
\begin{itemize}
\item Each client $c_i\in C$ can finally decode all the packets in $X$ from the packets sent by other clients via broadcast channel.
\item The total transmission cost $\delta$ defined in Eq.~(\ref{objective.cost}) is minimum among all the transmission schemes that satisfy the first condition.
\end{itemize}

Without loss of generality, we use $\overline{X_i}$ to denote the set of ``wanted" packets by client $c_i\in C$ , i.e., $\overline{X_i}=X\backslash X_i\subseteq X$. We also assume that the clients in $C$ are ordered with the non-decreasing order of the transmission cost, i.e., $\delta_{1}\leq \cdots \leq \delta_{N}$.

\section{Feasible transmission scheme}\label{Sec.solution}

Although the work in \cite{Love2007} already proposed a feasible transmission scheme, which can complete the third-party information exchange process with minimum transmissions, it is a special case of our problem, where it does not consider the transmission cost.
In this section, we aim to derive a necessary and sufficient condition for a feasible
transmission scheme, such that there exists a feasible code design for every client to successfully decode its ``wanted" packets. Then, based on the derived condition, we can give the transmission scheme to minimize the total transmission cost in Section~\ref{Sec.cost}.



\subsection{Encoding Matrix}
In this section, we define the encoding matrix of the transmitted packets.
Before sending the packets, each client first generates a linear encoded packet based on the packets it initially has over a finite field.
Then, the $k$-th encoded packet sent by client $c_i$ can be written as a linear combination of packets in $X_i$, i.e.,
\begin{eqnarray}
f^k_i=\sum_{j=1,j\neq i}^N \zeta^k_{i,j}x_{i,j}
\end{eqnarray}where $\zeta^k_{i,j}$ is the coefficient selected for packet $x_{i,j}$ by the $k$-th encoded packet of $c_i$ over finite field $GF(q)$.

The encoding vectors sent by all the clients can then be written as follows.
{\tiny\begin{eqnarray*}
E=\left(\begin{array}{ccccccccccc}
\zeta^1_{1,2} & \cdots &\zeta^{y_1}_{1,2} & \zeta^{1}_{2,1} & \cdots &\zeta^{y_2}_{2,1} & 0                 & \cdots & 0\\
\zeta^1_{1,3} & \cdots &\zeta^{y_1}_{1,3} & 0               & \cdots &0                 & \zeta^{1}_{3,1}   & \cdots &0\\
\cdot         & \cdots &\cdot             & \cdot           & \cdots &\cdot             & \cdot             & \cdots &\cdot\\
\cdot         & \cdots &\cdot             & \cdot           & \cdots &\cdot             & \cdot             & \cdots &\cdot\\
\cdot         & \cdots &\cdot             & \cdot           & \cdots &\cdot             & \cdot             & \cdots &\cdot\\
\zeta^1_{1,N} & \cdots & \zeta^{y_1}_{1,N}& 0               & \cdots &0                 &0                  & \cdots &\zeta^{y_N}_{N,1}\\
0             & \cdots & 0                & \zeta^{1}_{2,3} &\cdots  &\zeta^{y_2}_{2,3} &\zeta^{1}_{3,2}    & \cdots &0\\
0             & \cdots & 0                & \zeta^1_{2,4}   &\cdots  &\zeta^{y_2}_{2,4} & 0                 & \cdots &0\\
\cdot         & \cdots & \cdot            & \cdot           &\cdots  &\cdot             &\cdot              & \cdots &\cdot\\
\cdot         & \cdots & \cdot            & \cdot           &\cdots  &\cdot             &\cdot              & \cdots &\cdot\\
\cdot         & \cdots & \cdot            & \cdot           &\cdots  &\cdot             &\cdot              & \cdots &\cdot\\
0             & \cdots & 0                & \zeta^1_{2,N}   &\cdots  &\zeta^{y_2}_{2,N} & 0                 & \cdots &0\\
0             & \cdots & 0                & 0               &\cdots  &0                 &\zeta^{1}_{3,4}    & \cdots &\zeta^{y_N}_{N,2}\\
\cdot         & \cdots & \cdot            & \cdot           &\cdots  &\cdot             &\cdot              & \cdots &\cdot\\
\cdot         & \cdots & \cdot            & \cdot           &\cdots  &\cdot             &\cdot              & \cdots &\cdot\\
0             & \cdots & 0                & 0               &\cdots  &0                 &0                  & \cdots &\zeta^{y_N}_{N,N-1}\\
\end{array}\right)
\end{eqnarray*}}
\hspace{-0.08in}In the above encoding matrix $E$, each column vector denotes the encoding vector of a transmitted packet, and each row vector represents how a native packet is encoded in the transmitted packets. For example, the first column vector denotes the encoding vector of packet $f^1_1$ sent by client $1$, while in the first row vector, if the element is non-zero, it means the packet $x_{1,2}$ is participated in the encoded packet represented by that column.
Let $v^k_i$ be the encoding vector of the packet $f^k_i$, which is of size $\frac{N(N-1)}{2}$. For example, $v^1_1=(\zeta^1_{1,2},\zeta^1_{1,3},\cdots,\zeta^1_{1,N},0,\cdots,0)$.
Thus, $E=[v^1_1,\cdots, v^{y_1}_1,v^1_2,\cdots,v^{y_N}_N]^T$.




Without loss of generality, for each client, we define a local receiving matrix as follows.
\begin{DD}
The local {\em receiving matrix} of client $i$, named $R^i$, is defined as the sub-matrix of $E$, which includes almost all the rows and columns of $E$ except the followings:
\begin{itemize}
\item The rows, which represent the encoding status of native packets in $X_i$;
\item The columns, which represent the encoding vectors of packets sent by client $i$.
\end{itemize}
\end{DD}
Thus, a row vector of $R^i$ denotes how a ``wanted" packet of client $i$ is encoded in the received packets.

For example, $R^1$ does not include the first $N-1$ rows of $E$, as the first $N-1$ rows represent how the native packets in $X_1=\{x_{1,i'}|\forall i'\in\{2,\cdots,N\}\}$ participate in the received packets, and does not include the first $y_1$ columns of $E$, as these $y_1$ columns denotes the packets sent by client $1$.
Thus, $R^i$ is a $\frac{(N-1)(N-2)}{2}\times (\sum_{i'=1,i'\neq i}^Ny_{i'})$ matrix including the encoding vectors received by client $c_i$. We use $\beta_l$ to denote the $l$-th row vector of local receiving matrix $R^i$.

\subsection{Condition for a Receiving Matrix with Full Row Rank}
We first investigate the condition, under which there exists a code design for a receiving matrix with full row rank.

\begin{DD}
We define {\em coefficient element} as the element in a row encoding vector $v$, which is non-zero and is selected over $GF(q)$.
Let $Coef(v)$ be the set of columns in $v$ whose elements are coefficient elements.
\end{DD}
For example, $Coef(v^1_1)=\{1,2,3,\cdots,N-1\}$.

Let $R=(\beta_1, \beta_2,\cdots, \beta_m)^T$ be a general $m\times n$ receiving matrix, where $m\leq n$ and $\beta_i$ is the $i$-th row vector of $R$.
We then give the necessary and sufficient condition that there exists a code design to ensure the rank of $M$ is $m$ as follows.
\begin{lemma}\label{lemma_matrix}
There exists a code design such that the rank of the receiving matrix $R_{m\times n}$ is $m$, if and only if for any $r$ row vectors in $\{\beta_{i_1},\beta_{i_2},\cdots,\beta_{i_r}\}$, the size of $\bigcup_{j=1}^r Coef(\beta_{i_j})$ is at least $r$, where $1\leq i_j, r\leq m$.
\end{lemma}
\begin{proof}
We first prove the necessary condition, where we assume that there exists a code design such that the rank of the local receiving matrix $R_{m\times n}$ is $m$.

According to this assumption, for a matrix $R$, we can find at least a set of $m$ coefficient elements which are selected from different rows and different columns. In other words, for $m$ rows, the size of $\bigcup_{i=1}^m Coef(\beta_i)$ is $m$.

In addition, as the number of rows of matrix $R$ is $m$ and the rank of $R$ is $m$, each row vector should be linear independent with each other. Hence, it means for each sub-matrix of $R$ with $r$ rows, its rank is the number of rows it includes, i.e., $r$. So, in any $r$ rows $\{i_1,i_2,\cdots,i_r\}$, we can find at least a set of $r$ coefficient elements which are selected from $r$ different rows and $r$ different columns, i.e., $|\bigcup_{j=1}^r Coef(\beta_{i_j})|\geq r$. Thus, the necessary condition is proved.

We then prove the sufficient condition, where we assume that, for any $r$ row vectors of local receiving matrix $R$, $\{i_1,i_2,\cdots,i_r\}$, the size of $\bigcup_{j=1}^r Coef(\beta_{i_j})$ is at least $r$.

First, we consider the first row vector of $R$. There must be at least a coefficient element in row one, since $|Coef(\beta_{1})|\geq 1$. We select any of such columns, e.g., $l_1\in Coef(\beta_{1})$. Then, considering the second row vector of $R$, there must be at least a coefficient element, whose column number is not $l_1$, since $|Coef(\beta_{1})\bigcup Coef(\beta_{2})|\geq 2$. We then select such a column $l_2$ in $Coef(\beta_{1})\bigcup Coef(\beta_{2})$, where $l_2\neq l_1$. We repeat this process, and in each of the following rows, we will be able to find a coefficient element, whose column number has not been selected so far, since $|\bigcup_{j=1}^r Coef(\beta_{j})|\geq r$. Let $\{l_1,l_2,\cdots,l_m\}$ be the set of $m$ columns that have been selected.

Suppose that $R'_{m\times m}$ is a sub-matrix of $R$, where $R'$ includes $m$ column vectors of $R$ and the set of the indices of these $m$ columns is $\{l_1,l_2,\cdots,l_m\}$.

We can then design the feasible code as follows. Considering the elements in the $k$-th row vector of $R'$, only the coefficient element located in the $l_k$-th column is assigned with non-zero, while the other coefficient elements that are in other columns of row $k$ are assigned with zero.

According to the above coefficient assignment, the determinant of matrix $R'$ can be expressed as the product of $m$ non-zero elements from different rows and different columns, e.g., in the $k$-th row, the non-zero element located in column $l_k$ is selected. Since the determinant of matrix $R'$ is non-zero, the rank of $R'$ is thus $m$. Correspondingly, the rank of $R$ is also $m$. Thus, the sufficient condition is proved.

Hence, Lemma~\ref{lemma_matrix} is proved.
\end{proof}

\subsection{Necessary and Sufficient Condition for Feasible Transmission}\label{Sec.transmission.scheme}
In this section, we aim to find a feasible transmission scheme, such that there exists a code design for encoding matrix to ensure the ranks of all the local receiving matrices $R^i$s are full (i.e., $\frac{(N-1)(N-2)}{2}$), for $\forall i\in\{1,\cdots,N\}$.
To simplify the following presentation, we first define the following
\begin{DD}
Let $\{x_{i_1,i'_1},x_{i_2,i'_2},\cdots,x_{i_r,i'_r}\}$ be a $r$-subset of packets in $X$. We define $IDX( \{x_{i_1,i'_1},x_{i_2,i'_2},\cdots,x_{i_r,i'_r}\})$ as the indices set of the clients who hold at least one of packets in $\{x_{i_1,i'_1},x_{i_2,i'_2},\cdots,x_{i_r,i'_r}\}$.
\end{DD}
For example, for a $2$-subset $\{x_{1,2},x_{2,3}\}$, we can obtain that $IDX(\{x_{1,2},x_{2,3}\})=\{1,2,3\}$.

Before deriving the necessary and sufficient condition, we first prove the following lemma.
\begin{lemma}\label{lemma.kr}
For any $r$-subset $\{x_{i_1,i'_1},x_{i_2,i'_2},\cdots,x_{i_r,i'_r}\}$ of native packets in $X$, when $\binom{k-1}{2}+1\leq r\leq \binom{k}{2}$ for $\forall r,k\geq 1$, the size of $IDX( \{x_{i_1,i'_1},x_{i_2,i'_2},\cdots,x_{i_r,i'_r}\})$ is at least $k$.
\end{lemma}
\begin{proof}
Firstly, we consider the case when $r=\binom{k-1}{2}+1$. We can easily obtain that more than $k-1$ clients involve in the defined set. This is because, for any $k-1$ clients, the number of packets held by them, but not held by any other client is at most $\binom{k-1}{2}$. Thus, for $r=\binom{k-1}{2}+1$ packets, we still need at least another one client to include the extra packet. In other words, at least $k$ clients are needed, i.e., $|IDX(\{x_{i_1,i'_1},x_{i_2,i'_2},\cdots,x_{i_r,i'_r}\})|\geq k$.

We then consider $r=\binom{k}{2}$. As in the above case, more than $k-1$ clients involve in the defined set.
The worst case is that $r$ packets are only held by $k$ clients but not held by any other clients, e.g., packets in $\{x_{i_1,i_2},\cdots,x_{i_1,i_k},x_{i_2,i_3},\cdots,x_{i_2,i_k},\cdots,x_{i_{k-1},i_k}\}$ are only held by clients in $\{c_{i_1},c_{i_2},\cdots,c_{i_k}\}$.
In this case, only $k$ clients can involve these $r$ packets in their encoded packets, i.e., $|IDX(\{x_{i_1,i'_1},x_{i_2,i'_2},\cdots,x_{i_r,i'_r}\})|=k$.

When $\binom{k-1}{2}+1< r<\binom{k}{2}$, we can also similarly prove that at least $k$ clients are needed, by just considering $r$ packets in $\{x_{i_1,i_2},\cdots,x_{i_1,i_k},x_{i_2,i_3},\cdots,x_{i_2,i_k},\cdots,x_{i_{k-1},i_k}\}$.

Hence, the lemma is proved.

%
%

\end{proof}

Based on the above Lemmas, we then discuss the necessary and sufficient condition of the feasible transmission scheme for our third-party information exchange problem.
\begin{Theorem}\label{Theorem.sufficient}
For any client in $C$, there exists a code design such that it can decode all its ``wanted" packets, if and only if the total number of packets that any $k$ clients send is at least $\binom{k}{2}$. That is
\begin{eqnarray}
\sum_{t=1}^ky_{i_t}\geq \binom{k}{2},\forall 1\leq k< N
\end{eqnarray}where $\forall \{i_{1},i_{2},\cdots,i_{k}\}\subseteq\{1,2,\cdots,N\}$ and $i_{t}\neq i_{t'}$.
\end{Theorem}
\begin{proof}
To guarantee that client $c_{j}\in C$ can eventually decode its ``wanted" packets in $\overline{X_{j}}$, the rank of its local receiving matrix $R^{j}$ should be $\frac{(N-1)(N-2)}{2}$.


We first prove the necessary condition, where we assume that after receiving $y_1,\cdots,y_{j-1},y_{j+1},\cdots,y_N$ packets from clients $1,\cdots,j-1,j+1,\cdots,N$ respectively, there exists a code design such that client $j$ can decode its ``wanted" packets. In other words, there exists a code design such that the rank of matrix $R^{j}$ is $\frac{(N-1)(N-2)}{2}$.




According to Lemma~\ref{lemma_matrix}, to guarantee the rank of $R^{j}$ is $\frac{(N-1)(N-2)}{2}$, for any $r$ row vectors of $R^{j}$, we must have
\begin{eqnarray}\label{eq.i1}
|\bigcup_{i=1}^r Coef(\beta_{l_i})|\geq r
 \end{eqnarray}Note that each row vector denotes how a native packet is participated in the received encoded packets. In other words, $r$ row vectors represent $r$ native packets to participate in the encoded packets. According to Lemma~\ref{lemma.kr}, for any $r$-subset packets in $\{x_{i_1,i'_1},x_{i_2,i'_2},\cdots,x_{i_r,i'_r}\}$, we have $|IDX( \{x_{i_1,i'_1},x_{i_2,i'_2},\cdots,x_{i_r,i'_r}\})|\geq k$, when $\binom{k-1}{2}+1\leq r\leq \binom{k}{2}$.
In the worst case, for a $r$-subset of packets, e.g., $\{x_{i_1,i_2},\cdots,x_{i_1,i_k},x_{i_2,i_3},\cdots,x_{i_2,i_k},\cdots,x_{i_{k-1},i_k}\}$, we have {\small$|IDX(\{x_{i_1,i_2},\cdots,x_{i_1,i_k},x_{i_2,i_3},\cdots,x_{i_2,i_k},\cdots,x_{i_{k-1},i_k}\})|=k$}, where $r=\binom{k}{2}$.
That is, only $k$ clients can encode the packets in this $r$-subset into their encoded packets. Let $l_{t}$ be the index of the row vector that represents how the native packet $x_{i_{t},i'_{t}}$ in the above $r$-subset is participated in the received packets. Thus, $|\bigcup_{t=1}^r Coef(\beta_{l_{t}})|=\sum_{t=1}^ky_{i_t}$. According to Eq.~(\ref{eq.i1}), we have \begin{eqnarray}\label{eq.i2}\sum_{t=1}^ky_{i_t}\geq r=\binom{k}{2}\end{eqnarray}
which thus proves the necessary condition.

We then prove the sufficient condition, where we assume that for any $k$ clients, the total number of packets they send is at least $\binom{k}{2}$, which means $\sum_{t=1}^ky_{i_t}\geq \binom{k}{2}$, where $i_t\in \{1,\cdots, N\}$.

According to Lemma~\ref{lemma.kr}, we can obtain that for any $r$ native packets, at least $k$ clients (e.g., $\{i_1,i_2,\cdots,i_k\}$) can encode them in their sending packets, where $\binom{k-1}{2}+1\leq r\leq \binom{k}{2}$. Thus, for these $r$ rows, we can obtain that
\begin{eqnarray}
|\bigcup_{t=1}^r Coef(\beta_{l_t})|\geq \sum_{t=1}^ky_{i_t}
\end{eqnarray}
where $\{l_1,l_2\cdots,l_r\}$ is supposed to be the indices set of the row vectors representing the encoding status of these $r$ native packets.

According to the assumption, we have
\begin{eqnarray}
|\bigcup_{t=1}^r Coef(\beta_{l_t})|\geq \binom{k}{2}\geq r
\end{eqnarray}
In addition, since $\sum_{t=1}^ky_{i_t}\geq \binom{k}{2}$, we can obtain that
{\small\begin{eqnarray}
\sum_{i=1,i\neq j}^Ny_i\geq \binom{N-1}{2}=\frac{(N-1)(N-2)}{2}\end{eqnarray}}which means, the row number of $R^j$ is less than the column number of $R^j$.

Thus, the size of $\bigcup_{t=1}^r Coef(\beta_{l_t})$ is at least $r$, if for any $k$ clients, the total number of packets they sent is at least $\binom{k}{2}$. According to Lemma~\ref{lemma_matrix}, we can obtain that $R^j$ is with full row rank $\frac{(N-1)(N-2)}{2}$, which thus proves the sufficient condition.

Thus, we complete the proof of Theorem~\ref{Theorem.sufficient}.
\end{proof}

\section{Transmission Scheme with Minimum Transmission Cost}\label{Sec.cost}
In this section, we first formulate the problem of minimizing the total transmission cost as an integer linear programming. Based on the proposed transmission scheme, we analyze the performance that can be achieved with random linear network coding over $GF(q)$.

\subsection{Transmission Scheme with Minimum Cost}
Based on Section~\ref{Sec.transmission.scheme}, we can formulate the problem of minimizing the total transmission cost such that all clients can decode their ``wanted" packets, as an Integer Linear Programming (ILP) as follows.
\begin{eqnarray}\label{objective}
\min \sum_{i=1}^N\delta_iy_i
\end{eqnarray}
subject to
\begin{eqnarray}\label{constraint.1}
\sum_{t=1}^k y_{i_t}\geq \binom{k}{2}, \forall i_t\in\{1,\cdots,N\}, 1\leq k<N
\end{eqnarray}
Based on the above ILP, we can obtain the transmission scheme with the minimum total transmission cost.

We also prove the following theorem, which can be used to further simplify Constraint (\ref{constraint.1}) of the ILP.
\begin{Theorem}\label{Corollary.1}
Suppose that $\{y_1,y_2,\cdots,y_N\}$ is the optimal transmission scheme with the minimum total transmission cost.
We must have $y_1\geq y_2\geq \cdots\geq y_N$, when it is assumed that $\delta_{1}\leq \delta_2 \leq \cdots \leq \delta_{N}$.
\end{Theorem}
\begin{proof}
We omit the proof due to its simplicity.
\end{proof}

Based on the above theorem, we can conclude that the client with lower transmission cost needs to transmit more packets than the client with higher transmission cost.

\begin{Corollary}
The constraint (\ref{constraint.1}) of ILP can be reduced to
\begin{eqnarray}
\sum_{i=1}^{k}y_{N-i+1}&\geq& \binom{k}{2}, \forall k\in \{1,2,\cdots,N-1\}\label{constraint.2}\\
y_{i-1}&\geq& y_i,\forall k\in \{2,3,\cdots,N\}\label{constraint3}
\end{eqnarray}
when it is assumed that $\delta_1\leq \delta_2\leq\cdots\delta_N$.
\end{Corollary}
\begin{proof}
We now prove that with the Constraint (\ref{constraint.2}) and (\ref{constraint3}), Constraint (\ref{constraint.1}) can also be satisfied.

For any $k<N$, with constraint (\ref{constraint.2}), we can obtain that $y_{k-1}+y_{k}+\cdots,y_N\geq \binom{k}{2}$. From Theorem~\ref{Corollary.1}, we can easily obtain the constraint (\ref{constraint3}), i.e., $y_1\geq y_2\geq\cdots\geq y_N$. Thus, for any $k$ clients, the total number of packets they send must be no less than $y_{k-1}+y_{k}+\cdots,y_N$. That is, for any $\{i_1,i_2,\cdots,i_k\}\subseteq \{1,2,\cdots,N\}$, we have
\begin{eqnarray}
y_{i_1}+y_{i_2}+\cdots+y_{i_k}&\geq& y_{k-1}+y_{k}+\cdots+y_N\notag\\
&\geq& \binom{k}{2}
\end{eqnarray}where $i_1\neq i_2\neq \cdots\neq i_k$.

From the above equation, we can obtain that for any $k$ clients, where $1\leq k< N$, the total number of packets they need to send is at least $\binom{k}{2}$, which thus proved the above Corollary.
\end{proof}

\subsection{Illustration with an Example}
We consider a network with four clients as an example. Suppose that the transmission cost at each client is $\delta_1=1,\delta_2=2,\delta_3=3,\delta_4=4$. As shown in Fig.~\ref{Fig.table}, with our transmission scheme, the total transmission cost is $9$. On contrary, with the transmission scheme proposed in \cite{Love2007}, which aims to minimize the total number of transmissions, the transmission cost is $10$. In addition, we can easily check that with the code design in our scheme each client can decode its ``wanted" packets. Fig.~\ref{Fig.table} also verifies the result given in Theorem \ref{Corollary.1}, i.e., the clients with lower transmission costs should send more packets than the clients with higher transmission costs.

\begin{figure}[h]
\begin{center}
\includegraphics[height=36mm,width=85mm]{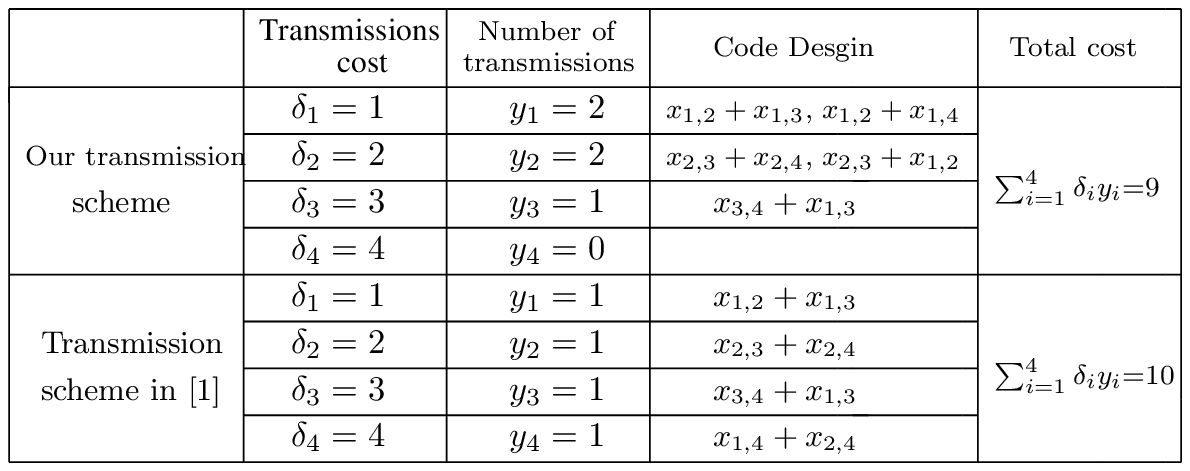}
\caption{Different transmission schemes with different total transmission costs.}\label{Fig.table}
\end{center}
\end{figure}

\subsection{Performance Analysis with Random Network Coding}

With the ILP in Eq.~(\ref{objective}) and Constraints~(\ref{constraint.2})~(\ref{constraint3}), we can obtain the optimal number of packets each client should send, so as to minimize the total transmission cost. To guarantee that each client $i$ can finally decode its ``wanted" packets with matrix $R^i$, we can design a deterministic code as introduced in Lemma~\ref{lemma_matrix}. However, the deterministic encoding matrix needs to be centrally designed, which may incur high overhead. Instead, we use random linear network coding at each client to locally determine the encoding vectors of the packets it sends.

We let each client $c_i\in C$ locally conduct random linear network coding over the packets that it initially has, where the number of encoded packets that each client should generate is determined by the ILP given in the above section.

Before analyzing further result, we introduce the
following Schwartz-Zippel Lemma~\cite{Motwani1995}.
\begin{lemma}\label{sch-zip}
(Schwartz-Zippel lemma~\cite{Motwani1995}) Let $P(z_1,\cdots,z_N)$ be a non-zero polynomial of degree $d\geq 0$ over a field $\mathbb F$. Let $S$ be a finite subset of $\mathbb F$, and the value of each $z_1,z_2,\cdots,z_N$ be selected independently and uniformly at random from $S$. Then the probability that the polynomial equals zero is at most $\frac{d}{|S|}$, i.e.,
$\Pr(P(z_1,\cdots,z_N)=0)\leq\frac{d}{|S|}$.
\end{lemma}

Based on the above lemma, we can derive the following probability.
\begin{Theorem}\label{Theorem.probability}
With random linear network coding and the transmission scheme $\{y_1,y_2,\cdots,y_N\}$ obtained by ILP, the probability that each client $i\in\{1,2,\cdots,N\}$ can finally decode its ``wanted" packets is at least
\begin{eqnarray}\label{Eq.probability}
1-\frac{(N-1)(N-2)}{2q}
\end{eqnarray} where $q$ is the field size.
\end{Theorem}
\begin{proof}
As in Theorem~\ref{Theorem.sufficient}, for any $k<N$ clients, the total number of packets they send is at least $\binom{k}{2}$. We then try to find a feasible set of the coefficients such that the local receiving matrix of each client $i$, $R^i$, is with rank $\frac{(N-1)(N-1)}{2}$.


For a matrix with maximum rank $\frac{(N-1)(N-2)}{2}$, the maximum degree of the coefficient variants should be $\frac{(N-1)(N-2)}{2}$. According to Lemma~\ref{sch-zip}, the probability that the determinant of this matrix is zero should be at most $\frac{(N-1)(N-2)}{2q}$. Hence, the probability that the determinant of the matrix is non-zero is at least
 $$1-\frac{(N-1)(N-2)}{2q}$$
where $q$ is the field size.

Thus, the probability that client $i$ can finally decode its ``wanted" packets with the local receiving matrix $R^i$ is at least $1-\frac{(N-1)(N-2)}{2q}$.
\end{proof}

Based on the above lemma, when the number of clients $N$ is fixed, we can increase the field size to enhance the probability that each client can finally decode its ``wanted" packets. The lower bound of the probability is shown in Table~\ref{table}.

\begin{table}[ht]
\caption{The probability lower bound in Theorem~\ref{Theorem.probability}} 
\centering 
\begin{tabular}{|c|c |c |c|c|c|} 
\hline 
 & {\em N=4} & {\em N=6} & {\em N=8} & {\em N=10} & {\em N=12} \\ 
 & {\em K=6} & {\em K=15}& {\em K=28}& {\em K=45} & {\em K=66}\\
\hline 
{\em q=256} & 0.9883 & 0.9609 &  0.9180 & 0.8594 &0.7852 \\ \hline
{\em q=512} & 0.9941 & 0.9805 &  0.9590 & 0.9297 & 0.8926 \\ \hline
\end{tabular}
\label{table} 
\end{table}

For example, when the total number of clients is $N=6$, which means the total number of packets needed to be exchanged is $K=\frac{N(N-1)}{2}=15$, the probability that each client can decode its ``wanted" packets is more than $98.05\%$, if we randomly select the coefficients from $q=512$.

\section{Conclusion}\label{Sec.conclusion}
In this paper, we aim to design a network coded cooperative information exchange scheme to minimize the total transmission cost for exchanging third-party information. We derive a necessary and sufficient condition for the feasible transmission scheme.
We prove that for any $k$ clients, where $1\leq k< N$, if the total number of packets they send is at least $\binom{k}{2}$,
there exists a feasible code design to make sure each client can finally obtain its ``wanted" packets. We further formulate the problem of minimizing the total transmission cost for third-party information exchange as an integer linear programming. Our analysis also shows that the clients with lower transmission cost should send more packets than the clients with higher transmission cost. Finally, based on the transmission scheme obtained by ILP, we provide a lower bound of the probability that each client can decode its ``wanted" packets, if random network coding is used.

\section{Acknowledgements}
This research is partly supported by the International Design Center (grant no. IDG31100102 \& IDD11100101). Li's work is partially supported by NSF under the Grants No CCF\-082988, CMMI\-0928092, and OCI\-1133027.

\end{document}